\documentclass[a4paper,12pt]{article} 
\usepackage[onehalfspacing]{setspace}%\usepackage[T1]{fontenc}
\usepackage{color}
\usepackage{enumitem}
\usepackage{amsmath,amssymb,amsthm}
\usepackage{subcaption}
\usepackage{romannum}
\usepackage{graphicx}
\usepackage{algorithm}
\usepackage{algpseudocode}
\usepackage{color}
\usepackage{anysize}
\usepackage{multirow}
\usepackage{authblk}
\newtheorem{theorem}{Theorem}[section]
\newtheorem{lemma}[theorem]{Lemma}

\theoremstyle{remark}

\title{Approximation Algorithms For The  Dispersion Problems in a Metric Space}

\author{Pawan K. Mishra\thanks{Department of Computer Science and Engineering,
        Indian Institute of Technology Guwahati}, Gautam K. Das \thanks{Department of Mathematics,   Indian Institute of Technology Guwahati}}

\begin{document}

\maketitle

	\begin{abstract}
		In this article, we consider the $c$-dispersion problem in a metric space $(X,d)$. 
		Let $P=\{p_{1}, p_{2}, \ldots, p_{n}\}$ be a set of $n$ points in a metric space $(X,d)$. For each point $p \in P$ and $S \subseteq P$, we define $cost_{c}(p,S)$ as the sum of distances from $p$ to the nearest $c $ points in $S \setminus \{p\}$, where $c\geq 1$ is a fixed  integer. We define $cost_{c}(S)=\min_{p \in S}\{cost_{c}(p,S)\}$ for  $S \subseteq P$. In the $c$-dispersion problem, a set $P$ of $n$ points in a metric space $(X,d)$ and a positive integer $k \in [c+1,n]$ are given. The objective is to find a subset $S\subseteq P$ of size $k$ such that $cost_{c}(S)$ is maximized.	
		
		 We propose a simple polynomial time greedy algorithm that produces a $2c$-factor approximation result for the $c$-dispersion problem in a metric space. The best known result for the $c$-dispersion problem in the Euclidean metric space $(X,d)$ is $2c^2$, where $P \subseteq \mathbb{R}^2$ and the distance function is Euclidean distance [ Amano, K. and Nakano, S. I., Away from Rivals, CCCG, pp.68-71, 2018 ].    We also prove that the $c$-dispersion problem in a metric space is $W[1]$-hard.

	\end{abstract}

	\section{Introduction}

	In facility location problem(FLP), we are given a set of $n$ locations on which some desired facilities to be placed and a positive integer $k$. The goal is to place  facilities  on $k$ locations out of $n$ given locations  such that specific objective is satisfied. Here,  specific objective depends on the nature of the problem. Suppose that the objective is to place $k$ facilities on $k$ locations, such that the closeness of chosen locations are undesirable. By closeness, we mean distance between a pair of facilities. We refer this FLP as a \emph{dispersion problem}. More specifically in the dispersion problem, we wish to minimize the interference between the placed facilities. One of  the most studied dispersion problems is the \emph{max-min dispersion problem}. 
	
	In the \emph{max-min dispersion problem},  we are given a set $ P = \{p_{1}, p_{2}, \ldots, p_{n}\}$ of $n$ locations,  the non-negative distance $d(p,q)$ between each pair of locations $p,q \in P$, and a positive integer $k$ ($k \leq n$). Here, $k$ refers to the number of facilities to be opened and distances are assumed to be symmetric. The objective is to find a $k$ size subset $S \subseteq P$  of locations such that $cost(S)= \min \{d(p,q) \mid p,q \in S \}$ is maximized.  This problem is known as $1$-dispersion problem in the literature.

 With reference to above mentioned problem, we extend the concept of closeness of a point $p_{i} \in S$ from one closest neighbor to given some specified number of closest  neighbor. Therefore, to preserve this notion of closeness of a point $p_{i}$, we need to add distances between   point $p_{i}$  and its specified number of nearest neighbors. We refer to such problem as $c$-dispersion problem.  Now, we define $c$-dispersion problem in a metric space $(X,d)$ as follows: 
	   
 \textbf{\emph{c-dispersion problem:}}  \emph{Let $P=\{p_{1}, p_{2}, \ldots, p_{n}\}$ be a set of $n$ points in a metric space $(X,d)$. For each point $p \in P$ and $S \subseteq P$, we define $cost_{c}(p,S)$ as the sum of   distances from $p$ to the first c nearest points in $S \setminus \{p\}$. We also define $cost_{c}(S)=\min_{p \in S}\{cost_{c}(p,S)\}$ for each $S \subseteq P$. In the $c$-dispersion problem, a set $P$ of $n$ points in a metric space $(X,d)$ and a positive integer $k \in [c+1,n]$ are given. The objective is to find a subset $S\subseteq P$ of size $k$ such that $cost_{c}(S)$ is maximized. }

In the real world, the dispersion problem has a huge number of applications. The possibility of opening chain stores in a community has piqued our interest in the dispersion problem. We need to open stores that are far apart from each other to eliminate/prevent self-competition among the stores. Installing dangerous facilities, such as nuclear power plants and oil tanks, is another condition in which dispersion is a concern. These facilities must be dispersed to the greatest extent possible, so that an accident at one site  does not affect others.  The dispersion problem is often used in information retrieval when we try to find a small subset of data with some desired variety from a large data set so that the small subset can be used as a valid sample to overview the large data set. 
	
\section{Related Work} \label{sec: relatedwork}
   
In 1977, Shier \cite{shier1977min} studied the two variants of the facility location problems on a tree network, namely the $k$-center problem and the max-min dispersion problem.  Sheir considered the continuum set of points on the rectifiable edges of a tree as a possible set of locations and showed that the max-min dispersion problem and $(k-1)$-center problem are dual. Sheir also established an equivalence between the max-min dispersion problem and $(k-1)$-center problem. In 1981, Chandrasekaran and Daughety \cite{chandrasekaran} studied the  max-min dispersion problem on a tree network and proposed a polynomial time algorithm.  In 1982, Chandrasekaran and Tamir \cite{chandrasekaran1982polynomially} also  studied the max-min dispersion problem and k-center problem on a tree network. They showed that if the set of locations is a finite subset of the continuum set of points on the rectifiable edges, then there exists an equivalence  between the max-min dispersion problem and  $(k-1)$-center problem on a tree network.  So, using  $k$-center algorithm on a tree (proposed in  \cite{frederickson1991optimal}),  a linear time algorithm for the max-min dispersion problem on a tree can be devised.  In \cite{erkut1990}, Erkut proved that the max-min dispersion problem is NP-hard even if the distance function satisfies  triangular  inequality. White \cite{white2} studied the max-min dispersion problem and proposed a  $3$-factor approximation result.  In 1991, Tamir \cite{tamir1991obnoxious} studied the max-min dispersion problem on a graph, where  continuum  set of points on the rectifiable edges are considered  as  locations.   Tamir  showed that for  a continuum set of locations on a graph, the max-min dispersion problem  can not be approximated within a factor of $3/2$ unless P=NP. In  \cite{tamir1991obnoxious},  a heuristic  is proposed  that  produces a $2$-factor approximation result for the max-min dispersion problem on a graph. Later in 1994, Ravi et al. \cite{ravi}  studied the max-min dispersion problem on a complete graph, where each edge is associated with a non-negative weight (distance).  They independently analyzed  the same heuristic  proposed in \cite{tamir1991obnoxious} (for the max-min dispersion problem on a complete graph), and showed that the same heuristic produces a $2$-approximation result for the complete graph. Furthermore, they also demonstrated that unless P=NP, the max-min dispersion problem on complete graph cannot be approximated within a factor of $2$ even if the distance function satisfies the triangular inequality. In 1991, Megiddo and Tamir \cite{megiddo1991p2} designed  an $O(k^2 \log^2 n)$ algorithm for the $k$-center problem on a line when points are sorted in an order.   Note that the same algorithm can be adapted to solve max-min dispersion problem on a line (points are not necessarily ordered) in polynomial time.  In 2007, Bhattacharya and Shi \cite{bhattacharya2007optimal} proposed a linear time algorithm for the $k$-center problem on a line, where points are not necessarily ordered. It is to be noted that this algorithm can be adapted to solve max-min dispersion problem on a line (points are not necessarily ordered) in polynomial time. In the geometric settings, the max-min dispersion problem was first introduced by  Wang and Kuo \cite{wang1988study}. They consider  the problem in a $d$-dimensional space with Euclidean distance function between two points. They  proposed a dynamic programming algorithm that solves the problem in $O(kn)$ time for $d=1$. They also  proved that for $d=2$, the problem is NP-hard.  Recently, in \cite{akagi}, the exact algorithm for the max-min dispersion problem was shown by establishing a  relationship with the maximum independent set problem. They proposed an $O(n^{wk/3} \log n)$ time algorithm, where $w < 2.373$. In \cite{akagi}, Akagi et al. also studied two special cases of the max-min dispersion problem where set of points  (i)  lies on a line, and (ii) lies on a circle. They proposed a polynomial time exact algorithm for both the cases.

The \emph{max-sum $k$-dispersion problem} is another popular variant of the dispersion problem. Here, the objective is to maximize the sum of distances between $k$ facilities.  Tamir \cite{tamir1991obnoxious} shown that the problem on a line has a trivial solution in $O(n)$ time. He also proved that the problem is solvable in $O(kn)$ time if the points are on a tree. Later, Ravi et al. \cite{ravi} also studied the problem on a line independently and gave a $O(\max(kn, n \log n))$ time algorithm. They also proposed a $4$-factor approximation algorithm if the distance function satisfies triangular inequality. In \cite{ravi}, they also proposed a $(1.571+ \epsilon)$-factor approximation algorithm for $2$-dimensional Euclidean space, where $\epsilon > 0$. In \cite{birnbaum} and \cite{hassin}, the approximation factor of $4$ was improved to $2$.  One can see \cite{baur2001approximation, chandra, kortsarz1993} for other variations of the dispersion problems.
	
In comparison with the max-min dispersion ($1$-dispersion) problem, a handful amount of research has been done in $c$-dispersion problem in a  metric space $(X,d)$.  Recently, in 2018, Amano and Nakano \cite{amano} proposed a greedy algorithm for the Euclidean  $2$-dispersion problem in $\mathbb{R}^2$, where the distance function between two points is the Euclidean distance. They have shown that the proposed greedy algorithm  produces an $8$-factor approximation result for the Euclidean $2$-dispersion problem in $\mathbb{R}^2$.  In 2020, \cite{amano2020} they analyzed the same greedy algorithm proposed in \cite{amano}, and shown that the greedy algorithm produces  a $4\sqrt{3}(\approx 6.92)$-factor approximation result for the  Euclidean $2$-dispersion problem. In \cite{amano}, they also proposed a $2c^2$ approximation result for the Euclidean $c$-dispersion problem in $\mathbb{R}^2$.

\subsection{Our Contribution}
In this article, we  consider the $c$-dispersion problem in a metric space and  propose a simple polynomial time  $2c$-factor approximation algorithm for a fixed $c$. We also proved that the $c$-dispersion problem in a metric space is $W[1]$-hard. 

\subsection{Organization of the Paper}
The remainder of the paper is organized as follows. In Section \ref{algorithm}, we propose a $2c$-factor approximation algorithm for the $c$-dispersion problem in a metric space. In Section \ref{hardness}, we prove that the $c$-dispersion problem in a metric space is $W[1]$-hard. We conclude the paper in Section \ref{Conclusion}.

\section{$2c$-Factor Approximation Algorithm for the $c$-Dispersion Problem in Metric Space}\label{algorithm}
In this section, we propose a  greedy algorithm for the $c$-dispersion problem in a metric space $(X,d)$. We will show that this algorithm guarantees  $2c$-factor approximation result for the $c$-dispersion problem in a metric space. Now, we discuss the greedy algorithm as follows. Let $I=(P, k)$ be an arbitrary instance of the $c$-dispersion problem in a metric space $(X,d)$, where $P=\{p_{1}, p_{2}, \ldots, p_{n}\}$  is  the set of $n$ points and $k \in [c+1,n]$ is a positive integer. It is an iterative algorithm. Initially, we choose a subset $S_{c+1} \subseteq S$ of size $c+1$ such that $cost_{c}(S_{c+1})$ is maximized. Next, we add one point  $p \in P$ into $S_{c+1}$ to construct $S_{c+2}$, i.e., $S_{c+2}=S_{c+1} \cup \{p\}$, such that $cost_{c}(S_{c+2})$ is maximized and continues this process up to the construction of $S_{k}$. The pseudo code of the algorithm is described in Algorithm \ref{GreedyDispersionAlgorithm}.  
	
\begin{algorithm}[H]
\caption{GreedyDispersionAlgorithm$(P, k)$}
\textbf{Input: } A set $P=\{p_1, p_2, \ldots, p_n\}$ of $n$ points, and a positive integer  $k (c+1 \leq k \leq n)$, along with distance function $d$.
		
\textbf{Output:} A subset $S_k \subseteq P$ of size $k$.	

\begin{algorithmic}[1]
			\State Compute $S_{c+1} = \{p_{i_1}, p_{i_2}, \ldots, p_{i_c}, p_{i_{c+1}}\} \subseteq P$ such that $cost_c(S_{c+1})$ is maximized. \label{line1algo1}
			\For {($j = c+2, c+3 , \ldots,  k$)} \label{line2algo}
			\State Let $p \in P\setminus S_{j-1}$ such that $cost_c(S_{j-1} \cup \{p\})$ is maximized. \label{impline}
			\State $S_j \leftarrow S_{j-1} \cup \{p\}$
			\EndFor
			\State return $(S_k)$
\end{algorithmic}  \label{GreedyDispersionAlgorithm}
\end{algorithm}
	
Let $S^*=\{p_{1}^*,p_{2}^*, \ldots, p_{k}^*\} \subseteq P$ be an optimum solution for the $c$-dispersion problem, i.e., $cost_{c}(S^*)=\max\limits_{{\substack {S \subseteq P \\ |S|=k}}} \{cost_{c}(S) \}$. 
We define a ball $B(p)$  at each point $p\in X$ as follows: $B(p) = \{q \in X |d(p, q) \leq \frac{cost_c(S^*)}{2c}\}$.  Let $B^*=\{B(p) \mid p \in S^*\}$. A point $q$ is \emph{properly contained} in $B(p)$, if $d(p,q) < \frac{cost_c(S^*)}{2c}$, whereas if $d(p,q) \leq \frac{cost_c(S^*)}{2c}$, then we say that point $q$ is \emph{contained} in $B(p)$.

\begin{lemma} \label{lemma01}
For any point $p \in P$, $B(p)$ can properly contains at most $c$ points from the optimal set $S^*$.
\end{lemma} 

\begin{proof}
On the  contrary assume that $B(p)$ properly contains $c+1$ points. Without loss of generality assume that    $p_{1}^*, p_{2}^*, \ldots ,p_{c+1}^* (\in S^*)$ are properly contained in $B(p)$. This implies that each of $d(p,p_1^*),d(p,p_2^*),  \ldots , d(p, p_{c+1}^*)$ is less than  $\frac{cost_c(S^*)}{2c}$. Since distance function $d(.,.)$ satisfies triangular inequality, $d(p_{1}^*,p_{j}^*) \leq d(p_{1}^*, p)+d(p,p_{j}^*)$, for $j= 2,3,\ldots,c+1$. This implies $\sum_{j=2}^{c+1}d(p_{1}^*,p_{j}^*) \leq \sum_{j=2}^{c+1} d(p_{1}^*,p) +\sum_{j=2}^{c+1} d(p,p_{j}^*) < 2c \times \frac{cost_c(S^*)}{2c} =cost_c(S^*)$. This leads to a contradiction that $p_1^* , p_2^*, \ldots, p_{c+1}^*  \in S^*$. Thus, the lemma.  
 \end{proof}
	 
\begin{lemma} \label{lemma02}
For any point $s \in P$,  if $ B' \subseteq  B^*$ is the set of balls that properly contains $s$, then $| B'| \leq c$.  
\end{lemma}
		
\begin{proof}
On the contrary assume that  $| B'| > c$. Without loss of generality assume that  $\{ B(p_1^*), B(p_2^*), \ldots , B(p_c^*),  B(p_{c+1}^*)\} \subseteq  B'$ are $c+1$ balls that properly contains  $s$. Here, $\{p_1^*,p_2^*, \ldots,p_c^*, p_{c+1}^*\} \subseteq S^*$. Since  $s$ is properly contained in $ B(p_1^*), B(p_2^*), \ldots, B(p_c^*), B(p_{c+1}^*)$, therefore each $d(p_{1}^*, s), d(p_{2}^*,s), \ldots, d(p_{c}^*, s),  d(p_{c+1}^*, s)$ is less than  $\frac{cost_{c}(S^*)}{2c}$. So, the ball $ B(s)=\{q  \mid d(s,q) \leq \frac{cost_{c}(S^*)}{2c}\}$ properly contains $c+1$ points $p_{1}^*, p_{2}^*, \ldots, p_{c}^*, p_{c+1}^*$ of the optimal set $S^*$, which is a contradiction to Lemma \ref{lemma01}. Thus, the lemma. 
\end{proof}	
	 
\begin{lemma}\label{lemma03}
Let $M \subseteq P$ be a set of points such that $|M|<k$. If $cost_{c}(M) \geq \frac{cost_{c}(S^*)}{2c}$, then there exists at least one ball $B(p_j^*) \in B^*=\{B(p_1^*),B(p_2^*), \ldots, B(p_k^*)\}$ that properly contains less than $c$ number of points in  $M$.
\end{lemma}	
	 
\begin{proof}
On the contrary assume that there does not exist any $j \in [1,k]$ such that $B(p_j^*)$ properly contains less than $c$ number of points in  $M$. Construct a bipartite graph $G(M \cup { B^*}, {\cal E})$ as follows: (i) $M$ and $ B^*$ are two partite vertex sets, and (ii) for $u \in M$, $(u, B(p_j^*)) \in {\cal E}$ if and only if $u$ is properly contained in $B(p_j^*)$.
	 	
According to assumption, each ball $B(p_j^*)$ properly contains at least $c$ points in $M$. Therefore, the total degree of the vertices in $ B^*$ in $G$ is at least $kc$. Note that $| B^*| = k$. On the other hand, the total degree of the vertices in $M$ in   $G$ is at most $c \times |M|$ (see Lemma \ref{lemma02}). Since $|M| < k$, the total degree of the vertices in $M$ in   $G$ is less than $ck$, which leads to a contradiction that the total degree of the vertices in $ B^*$ in  $G$ is at least $ck$. Thus, there exists at least one $p_{j}^* \in S^*$ such that  ball $B(p_j^*) \in B^*$  properly contains at most  $c-1$ points in $M$. 
\end{proof}

	 \begin{lemma}\label{prev_lemma}
	 The running time of Algorithm \ref{GreedyDispersionAlgorithm}  is $O(n^{c+1})$. 
	 \end{lemma}
	 
	 \begin{proof}
	     In line number \ref{line1algo1}, algorithm computes  $S_{c+1}$ such that $cost_{c}(S_{c+1})$ is maximized. To compute it,  algorithm  calculates  $cost_{c}(S_{c+1})$ for each distinct subset $S_{c+1} \subseteq P$ independently.  So, algorithm invests $ O{n \choose c+1 }= O(n^{c+1})$ time to compute $S_{c+1}$ such that $cost_{c}(S_{c+1})$ is maximized. Note that the value of $c$ is fixed. Now, for choosing a point in each iteration, algorithm takes $O(n^2)$ time. Here, the number of iteration is bounded by $k \leq n$.  So, to construct a set $S_{k}$ of size $k$ from $S_{c+1}$, algorithm takes $O(n^3)$ time.  Since  line number \ref{line1algo1} of Algorithm \ref{GreedyDispersionAlgorithm} takes a substantial amount of time compared to other steps of the algorithm, therefore the overall time complexity is  $O(n^{c+1})$.
	 \end{proof}
	 
	\begin{theorem}
		Algorithm \ref{GreedyDispersionAlgorithm} produces $2c$-factor approximation result in polynomial time for the $c$-dispersion problem.
	\end{theorem}
	
	\begin{proof}
		Let $I = (P, k)$ be an arbitrary input instance of the $c$-dispersion problem in a metric space $(X,d)$, where $P=\{p_1, p_2, \ldots, p_n\}$ is the set of $n$ points  and $k$ is a positive integer. Let $S_k$ and $S^*$ be an output of Algorithm \ref{GreedyDispersionAlgorithm} and an optimum solution, respectively, for the instance $I$. To prove the theorem, we  show that  $\frac{cost_c(S^*)}{cost_c(S_k)} \leq 2c$ . Here we use induction to show that $cost_c(S_i) \geq \frac{cost_c(S^*)}{2c}$ for each $i = c+1, c+2, \ldots, k$. 
		
		Since $S_{c+1}$ is an optimum solution for $c+1$ points (see line number \ref{line1algo1} of Algorithm \ref{GreedyDispersionAlgorithm}), therefore $cost_c(S_{c+1}) \geq cost_{c}(S^*) \geq  \frac{cost_c(S^*)}{2c}$ holds. Now, assume that the condition holds for  each $i$ such that $c+1 \leq i < k$. We will prove that the condition holds for $(i+1)$ too.
		
		Let $B^*$ be the set of balls corresponding to points in $S^*$. Since $i < k$ and $S_i \subseteq P$ with condition $cost_{c}(S_i) \geq \frac{cost_{c}(S^*)}{2c}$, then there exists at least one ball $B(p_j^*)\in B^*$ that properly contains at most $c-1$ points in $S_i$ (see Lemma \ref{lemma03}). Now, if $B(p_j^*)$  properly contains  $c-1$ points of the set $S_{i}$, then the distance of $p_{j}^*$ to the $c$-th closest point in $S_{i}$ is greater than or equal to $ \frac{cost_{c}(S^*)}{2c}$. Now, if we choose  a point $p_{j}^* \in S^*$ to the set $S_{i}$ to construct set $S_{i+1}$ (line number \ref{impline} of the algorithm), then $cost_{c}(p_j^*, S_{i+1}) \geq \frac{cost_{c}(S^*)}{2c}$. 
		Let $p \in S_{i+1}$ be an arbitrary point. Now, to prove $cost_{c}(p,S_{i+1}) \geq \frac{cost_{c}(S^*)}{2c}$, we consider following cases: (1) $p_j^*$ is not in the $c$-th nearest point of $p$ in $S_{i+1}$, and (2) $p_j^*$ is one of the $c$ nearest points of $p$ in $S_{i+1}$. In case (1), $cost_{c}(p,S_{i+1}) \geq \frac{cost_{c}(S^*)}{2c}$ by the definition of the set $S_{i}$, and in case (2) there exists at least one point $q \in S_{i+1}$ such that $d(q,p_j^*) \geq \frac{cost_{c}(S^*)}{2c}$, and $q$ is one of the $c$ nearest points of $p$ (see Lemma \ref{lemma02}). Therefore, sum of the distances of $p$ from $c$ nearest point in $S_{i+1}$ is greater than $d(p,p_j^*)+d(p,q) \geq d(p_j^*,q) \geq \frac{cost_{c}(S^*)}{2c}$. Therefore, we can conclude that if we consider the set $S_{i+1}=S_{i} \cup \{p_j^*\}$, then $cost_{c}(S_{i+1}) \geq \frac{cost_{c}(S^*)}{2c}$.
		  
		Since our algorithm chooses a point (see line number \ref{impline} of Algorithm \ref{GreedyDispersionAlgorithm}) that maximizes $cost_{c}(S_{i+1})$, therefore it will always choose a point in the iteration $i+1$ such that  $cost_{c}(S_{i+1}) \geq \frac{cost_{c}(S^*)}{2c }$.

		By the help of Lemma \ref{lemma01}, Lemma \ref{lemma02} and Lemma \ref{lemma03}, we have $cost_{c}(S_{i+1})\geq \frac{cost_{c}(S^*)}{2c }$ and thus condition holds for $(i+1)$ too.  Also, Lemma \ref{prev_lemma} says that Algorithm \ref{GreedyDispersionAlgorithm} computes $S_k$ in polynomial time. Therefore,  Algorithm \ref{GreedyDispersionAlgorithm} produces $2c$-factor approximation result in polynomial time for the $c$-dispersion problem.
		
	\end{proof}

	\section{ $c$-Dispersion Problem is W[1]-hard} \label{hardness} 
	In this section, we discuss the hardness of the $c$-dispersion problem in a metric spaces $(X,d)$ in the realm of parameterized complexity. We prove that the $c$-dispersion problem in a metric spaces $(X,d)$ is $W[1]$-hard. We show a  parameterized reduction from  $k$-independent set problem (known to be $W[1]$-hard \cite{flum2006parameterized}) to the $c$-dispersion problem in  $(X,d)$. 

 We  define parameterized version of both the problems as follows.\\
 
\noindent   
    \textbf{\emph{k-Independent Set Problem} } \\
	\emph{Instance:}  A graph $G=(V,E)$ and a positive integer $k$.\\
	\emph{Parameter:} $k$\\
	\emph{Problem:}  Does there exist an independent set of size $k$ in $G$ ? \\

\noindent	
	\textbf{\emph{$c$-Dispersion Problem } } \\
	\emph{Instance:}  A set  $P$ of $n$ locations and a positive integer $k$.\\
	\emph{Parameter:} $k$\\
	\emph{Problem:}  Given a bound $2c$, does there exist a subset $S_{k} \subseteq P$  of  size $k$ such that $cost_{c}(S_{k}) $ is  $2c$ ?  
	
	\begin{theorem}
		$c$-dispersion problem in a metric space $(X,d)$ is $W[1]$-hard. 
	\end{theorem}
	
	\begin{proof}
		We prove this by giving a parameterized  reduction from the $k$-independent set problem in simple undirected graphs to the $c$-dispersion problem in  $(X,d)$.  Now, we present a method to construct an instance of $c$-dispersion problem  from  any instance of the $k$-independent set problem in polynomial time.

		Let $G=(V,E)$ be an arbitrary instance of the $k$-independent set problem. Here, $V=\{v_1,v_2, \ldots,v_{n}\}$. We construct an instance of the $c$-dispersion problem from the given instance  of the  $k$-independent set problem.  We use set $V$ of vertices of $G$ as a set of locations $P$, i.e., $P=\{p_i \mid v_i \in V\}$  of $n$ points.  We define distance between points $p_{i}, p_{j}  \in P$ as follows: $d(p_{i}, p_{j})=1$ if $(v_{i},v_{j}) \in E$,  and   $d(p_{i}, p_{j})=2$, otherwise. Note that  this distance function satisfies triangle inequality.  So, the entire process of constructing  an instance  of the $c$-dispersion problem takes polynomial time.

		\textbf{Claim.} $G$ has independent set of size $k$ if and only if there exists a subset $S_{k} \subseteq P$ of size $k$, such that $cost_{c}(S_{k})=2c$.
		
		\noindent {\bf Necessity:}
		Let $I \subseteq V$ be an independent set of $G$ such that  $|I|=k$. We construct a set $S_k \subseteq P$ by selecting $k$ points in $P$ corresponding to  vertices in $I$, i.e., $S_k =\{p_i | v_i \in I\}$. Since, $I$ is an independent set, therefore by construction of an instance of the $c$-dispersion problem from $G$,  distance between any two points in  $S_{k}$ is $2$. This implies that   for each $p \in S_k$, $cost_{c}(p, S_{k})=2c$. Therefore, $cost_{c}(S_k)=2c$.

		\noindent {\bf Sufficiency:}
		Suppose there exists a  subset $S_{k} \subseteq P$, such that $cost_{c}(S_{k})=2c$.  Since $cost_{c}(S_{k})=2c$, this implies that there exists a point $p \in S_k$ such that $cost_{c}(p,S_k)=2c$ and for all $q \in S_k$,  $cost_{c}(q,S_k) \geq cost_{c}(p,S_k)$. Now if for a  point $q \in S_k$, $cost_{c}(q,S_k) > 2c$, then  by pigeon hole principle, distance of $q$ to  one of the $c$ nearest points in $S_k$ is greater than 2, which is not possible as per our construction of an instance of the $c$-dispersion problem. So,  for all points  $q \in S_k$, $cost_{c}(q,S_k)=2c$.  Now, we can create a set $I \subseteq V$ by selecting vertices   corresponding to each point in $S_{k}$, i.e., $I=\{v_i \mid p_i \in S_k\}$. Since distance between each pair of points is $2$,    therefore there does not exist any edge in $I$. Therefore, $I\subseteq V$ is an independent set of size $k$.

		Since $k$-independent set problem is $W[1]$-hard for a parameter $k$ \cite{flum2006parameterized},  and therefore using the above  reduction we conclude that  the $c$-dispersion problem in a metric space $(X,d)$ is also  $W[1]$-hard for the  same parameter $k$. Thus, the theorem.
	\end{proof}
	
\section{Conclusion }\label{Conclusion}
In this article, we studied the $c$-dispersion problem in a  metric space. We presented a polynomial time $2c$-factor  approximation algorithm for the $c$-dispersion problem in a metric space.  The best known approximation factor available for the Euclidean $c$-dispersion problem in $\mathbb{R}^2$ is $2c^2$ \cite{amano}. For $c = 1$, the proposed algorithm will produce a 2-factor approximation result, which is the same as the result in \cite{ravi, tamir1991obnoxious}. Therefore, our proposed algorithm is a generalized version and provides a better approximation result for the problem. We also proved that the $c$-dispersion problem in a  metric space is W[1]-hard.  

\section*{Acknowledgement}
We thank Prof. Arie Tamir (Tel Aviv University, Tel Aviv, Israel) for his useful suggestions on the literature of the problem.

%---------------------------- Bibliography -------------------------------

% Please add the contents of the .bbl file that you generate,  or add bibitem entries manually if you like.
% The entries should be in alphabetical order

\small
\bibliographystyle{abbrv}

\end{document}